\documentclass[10pt,journal,epsfig]{IEEEtran}
\usepackage[dvips]{graphicx}
\usepackage{graphicx}
\usepackage{amssymb}
\usepackage{cite}
\usepackage{amsmath}

\usepackage{subfigure}
\usepackage[centerlast]{caption2}

\usepackage{bm}

\begin{document}

\title{A New Determinant Inequality of Positive Semi-Definite Matrices}

\author{Jun Fang, ~\IEEEmembership{Member,~IEEE}, and Hongbin Li,~\IEEEmembership{Senior
Member,~IEEE}
\thanks{Jun Fang is with the National Key Laboratory on Communications,
University of Electronic Science and Technology of China, Chengdu
610054, China, Emails: Jun.Fang@ieee.org}
\thanks{Hongbin Li is with the Department of Electrical and
Computer Engineering, Stevens Institute of Technology, Hoboken, NJ
07030, USA, E-mail: Hongbin.Li@stevens.edu}
\thanks{This work was supported in part by the National Science
Foundation under Grant ECCS--0901066, the National Science
Foundation of China under Grant 61172114, and the Program for New
Century Excellent Talents in University (China) under Grant
NCET-09-0261. }}

\maketitle

\begin{abstract}
A new determinant inequality of positive semi-definite matrices is
discovered and proved by us. This new inequality is useful for
attacking and solving a variety of optimization problems arising
from the design of wireless communication systems.
\end{abstract}


\section{A New Determinant Inequality}
The following notations are used throughout this article. The
notations $[\cdot]^{T}$ and $[\cdot]^{H}$ stand for transpose and
Hermitian transpose, respectively. $\text{tr}(\boldsymbol{A})$ and
$\text{det}(\boldsymbol{A})$ denote the trace and the determinant of
the matrix $\boldsymbol{A}$, respectively. The symbols
$\mathbb{R}^{n\times m}$ and $\mathbb{R}^{n}$ stand for the set of
$n\times m$ matrices and the set of $n$-dimensional column vectors
with real entries, respectively. $\mathbb{C}^{n\times m}$ and
$\mathbb{C}^{n}$ denote the set of $n\times m$ matrices and the
set of $n$-dimensional column vectors with complex entries,
respectively.

We introduce the following new determinant inequality.
\newtheorem{theorem}{Theorem}
\begin{theorem} \label{theorem1}
Suppose $\boldsymbol{A}\in\mathbb{C}^{N\times N}$ and
$\boldsymbol{B}\in\mathbb{C}^{N\times N}$ are positive semi-definite
matrices with eigenvalues $\{\lambda_{k}(\boldsymbol{A})\}$ and
$\{\lambda_{k}(\boldsymbol{B})\}$ arranged in descending order,
$\boldsymbol{D}\in\mathbb{R}^{N\times N}$ is a diagonal matrix with
non-negative diagonal elements $\{d_k\}$ arranged in descending
order. Then the following determinant inequality holds
\begin{align}
\text{det}\left(\boldsymbol{D}^H\boldsymbol{A}\boldsymbol{D}+\boldsymbol{B}\right)\leq\prod_{k=1}^N
(d_k^2\lambda_{k}(\boldsymbol{A})+\lambda_{N+1-k}(\boldsymbol{B}))
\label{theorem1:eq1}
\end{align}
The above inequality becomes an equality if $\boldsymbol{A}$ and
$\boldsymbol{B}$ are diagonal, and the diagonal elements of
$\boldsymbol{A}$ and $\boldsymbol{B}$ are sorted in descending order and
ascending order, respectively, i.e.
$\boldsymbol{A}=\text{diag}(\lambda_{1}(\boldsymbol{A}),\ldots,\lambda_{N}(\boldsymbol{A}))$,
and
$\boldsymbol{B}=\text{diag}(\lambda_{N}(\boldsymbol{B}),\ldots,\lambda_{1}(\boldsymbol{B}))$.
\end{theorem}
\begin{proof}
See Appendix \ref{appA}.
\end{proof}

\section{Optimization Using The New Determinant Inequality}
The new determinant inequality can be used to solve the following
optimization problem
\begin{align}
\max_{\boldsymbol{C}}\quad & \det(\boldsymbol{G}+\boldsymbol{C}^H\boldsymbol{H}^H
\boldsymbol{R}_{v}^{-1}\boldsymbol{H}\boldsymbol{C})
\nonumber\\
\text{s.t.}\quad &
\text{tr}(\boldsymbol{C}\boldsymbol{R}_{x}\boldsymbol{C}^H)\leq P
\label{opt-3}
\end{align}
where $\boldsymbol{G}$, $\boldsymbol{R}_{v}$, and $\boldsymbol{R}_{x}$ are
positive definite matrices. Such an optimization arises when we
design the precoding matrix associated with a transmit node so as
to maximize the overall channel capacity (Details of the
formulation are omitted here).

To gain an insight into (\ref{opt-3}), we reformulate the problem
as follows. Let
$\boldsymbol{\bar{C}}\triangleq\boldsymbol{C}\boldsymbol{R}_{x}^{\frac{1}{2}}$,
and
$\boldsymbol{\bar{G}}\triangleq\boldsymbol{R}_{x}^{\frac{1}{2}}\boldsymbol{G}\boldsymbol{R}_{x}^{\frac{1}{2}}$,
the objective function (\ref{opt-3}) can be re-expressed as
\begin{align}
&\det(\boldsymbol{G}+\boldsymbol{C}^H\boldsymbol{H}^H
\boldsymbol{R}_{v}^{-1}\boldsymbol{H}\boldsymbol{C})\nonumber \\
=&\det(\boldsymbol{R}_{x}^{-1})\det(\boldsymbol{\bar{G}}+
\boldsymbol{\bar{C}}^H\boldsymbol{H}^H
\boldsymbol{R}_{v}^{-1}\boldsymbol{H}\boldsymbol{\bar{C}}) \label{equation-1}
\end{align}
To further simplify the problem, we carry out the SVD:
$\boldsymbol{\bar{C}}=\boldsymbol{U}_c\boldsymbol{D}_c\boldsymbol{V}_c^H$ and the
eigenvalue decomposition (EVD): $\boldsymbol{T}\triangleq\boldsymbol{H}^H
\boldsymbol{R}_{v}^{-1}\boldsymbol{H}=\boldsymbol{U}_t\boldsymbol{D}_t\boldsymbol{U}_t^H$,
and $\boldsymbol{\bar{G}}=\boldsymbol{U}_g\boldsymbol{D}_g\boldsymbol{U}_g^H$,
where $\boldsymbol{U}_c$, $\boldsymbol{V}_c$, $\boldsymbol{U}_t$, and
$\boldsymbol{U}_g$ are $p\times p$ unitary matrices, $\boldsymbol{D}_c$,
$\boldsymbol{D}_t$ and $\boldsymbol{D}_g$ are diagonal matrices given
respectively as
\begin{align}
\boldsymbol{D}_c\triangleq&\text{diag}(d_{c,1},d_{c,2},\ldots,d_{c,p})
\nonumber\\
\boldsymbol{D}_t\triangleq&\text{diag}(\lambda_1(\boldsymbol{T}),\lambda_2(\boldsymbol{T}),\ldots,\lambda_{p}(\boldsymbol{T}))
\nonumber\\
\boldsymbol{D}_g\triangleq&\text{diag}(\lambda_1(\boldsymbol{\bar{G}}),\lambda_2(\boldsymbol{\bar{G}}),\ldots,\lambda_{p}(\boldsymbol{\bar{G}}))
\label{definition-1}
\end{align}
in which $\lambda_k(\boldsymbol{T})$ and $\lambda_k(\boldsymbol{\bar{G}})$
denote the $k$-th eigenvalue associated with $\boldsymbol{T}$ and
$\boldsymbol{\bar{G}}$, respectively. Without loss of generality, we
assume that the diagonal elements of $\boldsymbol{D}_c$,
$\boldsymbol{D}_t$ and $\boldsymbol{D}_g$ are arranged in descending
order. We can rewrite (\ref{equation-1}) as
\begin{align}
&\det(\boldsymbol{R}_{x}^{-1})\det(\boldsymbol{\bar{G}}+
\boldsymbol{\bar{C}}^H\boldsymbol{H}^H
\boldsymbol{R}_{v}^{-1}\boldsymbol{H}\boldsymbol{\bar{C}}) \nonumber\\
=&\det(\boldsymbol{R}_{x}^{-1})\det(\boldsymbol{V}_c^H\boldsymbol{\bar{G}}\boldsymbol{V}_c+
\boldsymbol{D}_c^H\boldsymbol{U}_c^H\boldsymbol{T}\boldsymbol{U}_c\boldsymbol{D}_c)
\nonumber\\
\triangleq&\det(\boldsymbol{R}_{x}^{-1})\det(\boldsymbol{\bar{V}}_c^H\boldsymbol{D}_g\boldsymbol{\bar{V}}_c+
\boldsymbol{D}_c^H\boldsymbol{\bar{U}}_c^H\boldsymbol{D}_t\boldsymbol{\bar{U}}_c\boldsymbol{D}_c)
\label{equation-2}
\end{align}
where $\boldsymbol{\bar{V}}_c\triangleq\boldsymbol{U}_g^H\boldsymbol{V}_c$,
and $\boldsymbol{\bar{U}}_c\triangleq\boldsymbol{U}_t^H\boldsymbol{U}_c$.
Resorting to (\ref{equation-2}), the optimization (\ref{opt-3})
can be transformed into a new optimization that searches for an
optimal set
$\{\boldsymbol{\bar{U}}_c,\boldsymbol{D}_c,\boldsymbol{\bar{V}}_c\}$, in which
$\boldsymbol{\bar{U}}_c$ and $\boldsymbol{\bar{V}}_c$ are also unitary
matrices
\begin{align}
\max_{\{\boldsymbol{\bar{U}}_c,\boldsymbol{D}_c,\boldsymbol{\bar{V}}_c\}}\quad
& \det(\boldsymbol{\bar{V}}_c^H\boldsymbol{D}_g\boldsymbol{\bar{V}}_c+
\boldsymbol{D}_c^H\boldsymbol{\bar{U}}_c^H\boldsymbol{D}_t\boldsymbol{\bar{U}}_c\boldsymbol{D}_c)
\nonumber\\
\text{s.t.}\quad & \text{tr}(\boldsymbol{D}_c\boldsymbol{D}_c^H)\leq P_n
\nonumber\\
& \boldsymbol{\bar{U}}_c\boldsymbol{\bar{U}}_c^H=\boldsymbol{I}, \qquad
\boldsymbol{\bar{V}}_c\boldsymbol{\bar{V}}_c^H=\boldsymbol{I} \label{opt-7}
\end{align}
The optimization involves searching for multiple optimization
variables. Nevertheless, we can, firstly, find the optimal
$\{\boldsymbol{\bar{U}}_c,\boldsymbol{\bar{V}}_c\}$ given that
$\boldsymbol{D}_c$ is fixed. Then substituting the derived optimal
unitary matrices into (\ref{opt-7}), we determine the optimal
diagonal matrix $\boldsymbol{D}_c$. Optimizing
$\{\boldsymbol{\bar{U}}_c,\boldsymbol{\bar{V}}_c\}$ conditional on a given
$\boldsymbol{D}_c$ can be formulated as
\begin{align}
\max_{\{\boldsymbol{\bar{U}}_c,\boldsymbol{\bar{V}}_c\}}\quad &
\det(\boldsymbol{\bar{V}}_c^H\boldsymbol{D}_g\boldsymbol{\bar{V}}_c+
\boldsymbol{D}_c^H\boldsymbol{\bar{U}}_c^H\boldsymbol{D}_t\boldsymbol{\bar{U}}_c\boldsymbol{D}_c)
\nonumber\\
\text{s.t.}\quad &
\boldsymbol{\bar{U}}_c\boldsymbol{\bar{U}}_c^H=\boldsymbol{I}, \qquad
\boldsymbol{\bar{V}}_c\boldsymbol{\bar{V}}_c^H=\boldsymbol{I} \label{opt-8}
\end{align}

Letting
$\boldsymbol{A}\triangleq\boldsymbol{\bar{U}}_c^H\boldsymbol{D}_t\boldsymbol{\bar{U}}_c$,
$\boldsymbol{B}\triangleq\boldsymbol{\bar{V}}_c^H\boldsymbol{D}_g\boldsymbol{\bar{V}}_c$,
and utilizing Theorem \ref{theorem1}, the objective function
(\ref{opt-8}) is upper bounded by
\begin{align}
&\det(\boldsymbol{\bar{V}}_c^H\boldsymbol{D}_g\boldsymbol{\bar{V}}_c+
\boldsymbol{D}_c^H\boldsymbol{\bar{U}}_c^H\boldsymbol{D}_t\boldsymbol{\bar{U}}_c\boldsymbol{D}_c)\nonumber\\
&\leq\prod_{k=1}^{p}\left(d_{c,k}^2\lambda_k(\boldsymbol{A})+\lambda_{p+1-k}(\boldsymbol{B})\right)
\nonumber\\
&=\prod_{k=1}^{p}\left(d_{c,k}^2\lambda_k(\boldsymbol{T})+\lambda_{p+1-k}(\boldsymbol{\bar{G}})\right)
\end{align}
The above inequality becomes an equality when
$\boldsymbol{\bar{U}}_c=\boldsymbol{I}$ and
$\boldsymbol{\bar{V}}_c=\boldsymbol{J}$, where $\boldsymbol{J}$ is an
anti-identity matrix, that is, $\boldsymbol{J}$ has ones along the
anti-diagonal and zeros elsewhere. Therefore the optimal solution
to (\ref{opt-8}) is given by
\begin{align}
\boldsymbol{\bar{U}}_c=\boldsymbol{I}, \qquad
\boldsymbol{\bar{V}}_c=\boldsymbol{J}
\end{align}
Substituting the optimal
$\{\boldsymbol{\bar{U}}_c,\boldsymbol{\bar{V}}_c\}$ back into
(\ref{opt-7}), we arrive at the following optimization that
searches for optimal diagonal elements $\{d_{c,k}\}$
\begin{align}
\max_{\{d_{c,k}\}}\quad&
\prod_{k=1}^{p}\left(d_{c,k}^2\lambda_k(\boldsymbol{T})+\lambda_{p_n+1-k}(\boldsymbol{\bar{G}})\right)
\nonumber\\
\text{s.t.}\quad&\sum_{k=1}^{p}d_{c,k}^2\leq P, \quad d_{c,k}\geq
0 \quad\forall k \label{opt-9}
\end{align}
The above optimization (\ref{opt-9}) can be solved analytically by
resorting to the Lagrangian function and KKT conditions, whose
details are not elaborated here.

\useRomanappendicesfalse
\appendices

\section{Proof of Theorem \ref{theorem1}} \label{appA}
Define
$\boldsymbol{\mathit{\Gamma}}\triangleq\boldsymbol{D}^H\boldsymbol{A}\boldsymbol{D}$,
and its eigenvalues $\{\lambda_{k}(\boldsymbol{\mathit{\Gamma}})\}$ are
arranged in descending order. Then we have
\begin{align}
\text{det}\left(\boldsymbol{D}^H\boldsymbol{A}\boldsymbol{D}+\boldsymbol{B}\right)\leq
\prod_{k=1}^N(\lambda_{k}(\boldsymbol{\mathit{\Gamma}})+\lambda_{N+1-k}(\boldsymbol{B}))
\label{appA:eq1}
\end{align}
The above inequality comes from the following well-known matrix
inequality \cite{HornJohnson85}:
\begin{align}
\prod_{k=1}^N(\lambda_{k}(\boldsymbol{X})+\lambda_{k}(\boldsymbol{Y}))\leq&\det(\boldsymbol{X}+\boldsymbol{Y})
\nonumber\\
&\leq\prod_{k=1}^N(\lambda_{k}(\boldsymbol{X})+\lambda_{N+1-k}(\boldsymbol{Y}))
\end{align}
in which $\boldsymbol{X}$ and $\boldsymbol{Y}$ are positive semidefinite
Hermitian matrices, with eigenvalues $\{\lambda_{k}(\boldsymbol{X})\}$
and $\{\lambda_{k}(\boldsymbol{Y})\}$ arranged in descending order
respectively.

To prove (\ref{theorem1:eq1}), we only need to show that the term
on the right-hand side of (\ref{appA:eq1}) is upper bounded by
\begin{align}
\prod_{k=1}^N(\lambda_{k}(\boldsymbol{\mathit{\Gamma}})+\lambda_{N+1-k}(\boldsymbol{B}))\leq
\prod_{k=1}^N
(d_k^2\lambda_{k}(\boldsymbol{A})+\lambda_{N+1-k}(\boldsymbol{B}))
\label{appA:eq2}
\end{align}
Before proceeding to prove (\ref{appA:eq2}), we introduce the
following inequalities for the two sequences
$\{\lambda_{k}(\boldsymbol{\mathit{\Gamma}})\}_{k=1}^N$ and
$\{d_k^2\lambda_{k}(\boldsymbol{A})\}_{k=1}^N$.
\begin{align}
\prod_{k=1}^K\lambda_{k}(\boldsymbol{\mathit{\Gamma}})\leq&\prod_{k=1}^Kd_k^2\lambda_{k}(\boldsymbol{A}) \quad 1\leq K<N \nonumber\\
\prod_{k=1}^N\lambda_{k}(\boldsymbol{\mathit{\Gamma}})=&\prod_{k=1}^Nd_k^2\lambda_{k}(\boldsymbol{A})
\label{appA:eq3}
\end{align}
The proof of the inequalities (\ref{appA:eq3}) is provided in
Appendix \ref{appB}. The inequality relations between these two
sequences can be characterized by the notion of ``multiplicative
majorization'' (also termed log-majorization). Multiplicative
majorization is a notion parallel to the concept of additive
majorization. For two vectors $\boldsymbol{a}\in\mathbb{R}_{+}^N$ and
$\boldsymbol{b}\in\mathbb{R}_{+}^N$ with elements sorted in descending
order ($\mathbb{R}_{+}$ stands for the set of non-negative real
numbers), we say that $\boldsymbol{a}$ is multiplicatively majorized
by $\boldsymbol{b}$, denoted by $\boldsymbol{a}\prec_{\times}\boldsymbol{b}$,
if
\begin{align}
\prod_{k=1}^K a_k\leq& \prod_{k=1}^K b_k \quad 1\leq K<N
\nonumber\\
\prod_{k=1}^N a_k=&\prod_{k=1}^N b_k \label{appA:inequalities}
\end{align}
Here we use the symbol $\prec_{\times}$ to differentiate the
multiplicative majorization from the conventional additive
majorization $\prec$. Another important concept that is closely
related to majorization is schur-convex or schur-concave
functions. A function $f:\mathbb{R}^N\rightarrow\mathbb{R}$ is
said to be multiplicatively schur-convex if for
$\boldsymbol{a}\prec_{\times}\boldsymbol{b}$, then $f(\boldsymbol{a})\leq
f(\boldsymbol{b})$. Clearly, establishing (\ref{appA:eq2}) is
equivalent to showing the function
\begin{align}
f(\boldsymbol{a})\triangleq\prod_{k=1}^N(a_k+c_{N+1-k})
\end{align}
is multiplicatively schur-convex for elements
$\boldsymbol{c}=[c_k]\in\mathbb{R}^{N}_{+}$ arranged in descending
order. This multiplicatively schur-convex property can also be
summarized as follows.

\newtheorem{lemma}{Lemma}
\begin{lemma} \label{lemma1}
For vectors $\boldsymbol{a}\in\mathbb{R}_{+}^N$,
$\boldsymbol{b}\in\mathbb{R}_{+}^N$, and
$\boldsymbol{c}\in\mathbb{R}_{+}^N$, with their elements arranged in
descending order, if $\boldsymbol{a}\prec_{\times}\boldsymbol{b}$, then we
have $f(\boldsymbol{a})\leq f(\boldsymbol{b})$, i.e.
\begin{align}
\prod_{k=1}^N(a_k+c_{N+1-k})\leq \prod_{k=1}^N (b_k+c_{N+1-k})
\label{lemma1:eq1}
\end{align}
\end{lemma}
\begin{proof}
We prove Lemma \ref{lemma1} by induction. For $N=2$, we have
\begin{align}
f(\boldsymbol{a})-f(\boldsymbol{b})=&[a_1+c_2][a_2+c_1]
-[b_1+c_2][b_2+c_1]
\nonumber\\
\stackrel{(a)}{=}&[a_1-b_1]c_1+ [a_2-b_2]c_2
\nonumber\\
\stackrel{(b)}{\leq}&c_2[a_1+a_2 -b_1 -b_2] \stackrel{(c)}{\leq} 0
\label{appA:argument}
\end{align}
where $(a)$ can be easily derived by noting that $a_1a_2=b_1b_2$;
$(b)$ comes from the fact that $a_1-b_1\leq 0$ and $c_1\geq c_2$;
$(c)$ is a result of the following inequality: $a_1+a_2\leq
b_1+b_2$, that is, for any two non-negative elements, if their
product remains constant, then their sum increases as the two
elements are further apart.

Now suppose that for $M$-dimensional vectors $\boldsymbol{a}$,
$\boldsymbol{b}$ and $\boldsymbol{c}$, the inequality
(\ref{lemma1:eq1}) holds true. We show that (\ref{lemma1:eq1}) is
also valid for $(M+1)$-dimensional vectors $\boldsymbol{a}$,
$\boldsymbol{b}$ and $\boldsymbol{c}$. From the inequalities
(\ref{appA:inequalities}), we know that $b_1\geq a_1$. For the
special case where $b_1=a_1$, it is easy to verify that the
truncated vector $\boldsymbol{a}_t\triangleq
[a_2\phantom{0}\ldots\phantom{0}a_{M+1}]$ is multiplicatively
majorized by the truncated vector $\boldsymbol{b}_t\triangleq
[b_2\phantom{0}\ldots\phantom{0}b_{M+1}]$, i.e.
$\boldsymbol{a}_t\prec_{\times}\boldsymbol{b}_t$. Therefore we
have
\begin{align}
f(\boldsymbol{a}_t)\leq f(\boldsymbol{b}_t)
\end{align}
and consequently we arrive at $f(\boldsymbol{a})\leq f(\boldsymbol{b})$
given $b_1=a_1$.

Now consider the general case where $b_1>a_1$. There must be at
least one index such that $b_l<a_l$ since the overall products of
the two sequences $\{a_k\}_{k=1}^{M+1}$ and $\{b_k\}_{k=1}^{M+1}$
are identical\footnote{When $\boldsymbol{a}$ and $\boldsymbol{b}$
contain zero elements, the overall product is zero. In this case,
we may not find an index such that $b_l<a_l$. Nevertheless, since
we have $b_k\geq a_k$ for all $k$, proof of (\ref{lemma1:eq1}) is
evident.}. Without loss of generality, let $l_1$ denote the
smallest index for which $b_l<a_l$. We adopt a pairwise
transformation to convert the sequence $\{b_l\}_{k=1}^{M+1}$ into
a new sequence $\{\beta_k\}_{k=1}^{M+1}$. Specifically, the first
and the $l_1$th entries of $\{b_l\}_{k=1}^{M+1}$ are updated as
\begin{equation}
\begin{cases} \beta_1=a_1, \beta_{l_1}=
\frac{b_1b_{l_1}}{a_1} & \text{if } b_1b_{l_1}
\leq a_1a_{l_1}\\
\beta_1=\frac{b_1b_{l_1}}{a_{l_1}} , \beta_{l_1}=a_{l_1}& \text{if
} b_1b_{l_1}
>a_1a_{l_1}
\end{cases} \label{appA:eq5}
\end{equation}
whereas other entries remain unaltered, i.e. $\beta_k=b_k$,
$\forall k\neq 1,l_1$. Clearly, the entries $\beta_1$ and
$\beta_{l_1}$ satisfy
\begin{align}
\beta_1\beta_{l_1}=&b_1b_{l_1} \nonumber\\
\beta_1\leq &b_1 \label{appA:eq4}
\end{align}
That is,
$[\beta_1\phantom{0}\beta_{l_1}]^T\prec_{\times}[b_1\phantom{0}b_{l_1}]^T$.
By following the same argument of (\ref{appA:argument}) and noting
that $\beta_k=b_k$, $\forall k\neq 1,l_1$, we have
\begin{align}
f(\boldsymbol{\beta})\leq f(\boldsymbol{b}) \label{appA:eq6}
\end{align}
where
$\boldsymbol{\beta}\triangleq[\beta_1\phantom{0}\ldots\phantom{0}\beta_{M+1}]$.
Our objective now is to show
\begin{align}
f(\boldsymbol{a})\leq f(\boldsymbol{\beta}) \label{appA:eq8}
\end{align}
It can be easily verified that $\boldsymbol{a}$ is multiplicatively
majorized by $\boldsymbol{\beta}$, i.e.
$\boldsymbol{a}\prec_{\times}\boldsymbol{\beta}$, by noting
$\beta_l\geq a_l$ for any $l<l_1$ and
$\beta_1\beta_{l_1}=b_1b_{l_1}$.

Now we proceed to prove (\ref{appA:eq8}). Consider two different
cases in (\ref{appA:eq5}).
\begin{itemize}
\item If $b_1b_{l_1} \leq a_1a_{l_1}$, then $\beta_1=a_1$. In this
case, it is easy to verify that the truncated vector
$\boldsymbol{a}_t\triangleq [a_2\phantom{0}\ldots\phantom{0}a_{M+1}]$
is multiplicatively majorized by the truncated vector
$\boldsymbol{\beta}_t\triangleq
[\beta_2\phantom{0}\ldots\phantom{0}\beta_{M+1}]$, i.e.
$\boldsymbol{a}_t\prec_{\times}\boldsymbol{\beta}_t$. Therefore we
have
\begin{align}
f(\boldsymbol{a}_t)\leq f(\boldsymbol{\beta}_t)
\end{align}
and consequently $f(\boldsymbol{a})\leq f(\boldsymbol{\beta})$ as we
have $\beta_1=a_1$.
\item For the second case where $b_1b_{l_1}>a_1a_{l_1}$, we have
$\beta_{l_1}=a_{l_1}$. Define two new vectors
$\boldsymbol{a}_p\triangleq
[a_1\phantom{0}\ldots\phantom{0}a_{l_1-1}\phantom{0}a_{l_1+1}\phantom{0}\ldots,a_{M+1}]$
and $\boldsymbol{\beta}_p\triangleq
[\beta_1\phantom{0}\ldots\phantom{0}\beta_{l_1-1}\phantom{0}\beta_{l_1+1}\phantom{0}\ldots\phantom{0}\beta_{M+1}]$.
From $\boldsymbol{a}\prec_{\times}\boldsymbol{\beta}$, we can readily
verify that $\boldsymbol{a}_p$ is multiplicatively majorized by
$\boldsymbol{\beta}_p$, i.e.
$\boldsymbol{a}_p\prec_{\times}\boldsymbol{\beta}_p$. Therefore we
have
\begin{align}
f(\boldsymbol{a}_p)\leq f(\boldsymbol{\beta}_p)
\end{align}
and consequently $f(\boldsymbol{a})\leq f(\boldsymbol{\beta})$ as we
have $\beta_{l_1}=a_{l_1}$.
\end{itemize}

Combining (\ref{appA:eq6})--(\ref{appA:eq8}), we arrive at
(\ref{lemma1:eq1}). The proof is completed here.
\end{proof}


\section{Proof of (\ref{appA:eq3})} \label{appB}
Recall the following theorem \cite[Chapter 9: Theorem
H.1]{MarshallOlkin79}

\emph{Theorem:} If $\boldsymbol{X}$ and $\boldsymbol{Y}$ are $N\times N$
complex matrices, then
\begin{align}
\prod_{k=1}^K\sigma_k(\boldsymbol{X}\boldsymbol{Y})\leq&\prod_{k=1}^K\sigma_k(\boldsymbol{X})\sigma_k(\boldsymbol{Y}),
\qquad K=1,\ldots,N-1 \nonumber\\
\prod_{k=1}^N\sigma_k(\boldsymbol{X}\boldsymbol{Y})=&\prod_{k=1}^N\sigma_k(\boldsymbol{X})\sigma_k(\boldsymbol{Y})
\end{align}
where $\{\sigma_i(\cdot)\}$ are singular values arranged in a
descending order.

By utilizing the above results, we have
\begin{align}
\prod_{k=1}^K\lambda_{k}(\boldsymbol{\mathit{\Gamma}})=&\prod_{k=1}^K\sigma_{k}(\boldsymbol{\mathit{\Gamma}})
\leq\prod_{k=1}^K\sigma_{k}(\boldsymbol{D}^H\boldsymbol{A})\sigma_{k}(\boldsymbol{D})
\nonumber\\
=&\bigg(\prod_{k=1}^K\sigma_{k}(\boldsymbol{D}^H\boldsymbol{A})\bigg)\bigg(\prod_{k=1}^K
d_k\bigg) \nonumber\\ \leq&
\bigg(\prod_{k=1}^K\sigma_{k}(\boldsymbol{D}^H)\sigma_{k}(\boldsymbol{A})\bigg)
\bigg(\prod_{k=1}^K d_{k}\bigg) \nonumber\\
=&\prod_{k=1}^K d_k^2\lambda_{k}(\boldsymbol{A}), \qquad
K=1,\ldots,N-1
\end{align}

\end{document}